\documentclass[11pt]{article}%
\usepackage{graphicx}

\usepackage{fancyhdr}
\usepackage[hidelinks]{hyperref}
\usepackage[a4paper, margin = 1in]{geometry}
\usepackage{hyperref}
\usepackage{calrsfs}

\usepackage{times}
\usepackage{amsmath}
\usepackage{changepage}
\usepackage{amssymb}
\usepackage{graphicx}%
\newtheorem{theorem}{Theorem}[section]
\newtheorem{corollary}[theorem]{Corollary}
\newtheorem{definition}[theorem]{Definition}
\newtheorem{lemma}[theorem]{Lemma}

\newtheorem{observation}[theorem]{Observation}
\newtheorem{claim}[theorem]{Claim}
\newtheorem{inf_theorem}[theorem]{Informal Theorem}

\newtheorem{remark}[theorem]{Remark}

\usepackage{mathtools}

\newenvironment{proof}[1][Proof]{\textbf{#1.} }{\ \rule{0.5em}{0.5em}}
\usepackage{breqn}
\usepackage{physics}
\usepackage{bbm}

\usepackage{cases}    

\begin{document}


\title{Tight Bounds on 3-Team Manipulations in Randomized Death Match}

\author{
  Atanas Dinev\\
  Massachusetts Institute of Technology, Cambridge MA, USA
  \and
  S. Matthew Weinberg\\
  Princeton University, Princeton, USA\\
}

\date{\today}
\maketitle

\begin{abstract}

Consider a round-robin tournament on $n$ teams, where a winner must be (possibly randomly) selected as a function of the results from the $\binom{n}{2}$ pairwise matches. A tournament rule is said to be $k$-SNM-$\alpha$ if no set of $k$ teams can ever manipulate the $\binom{k}{2}$ pairwise matches between them to improve the joint probability that one of these $k$ teams wins by more than $\alpha$. Prior work identifies multiple simple tournament rules that are $2$-SNM-$1/3$ (Randomized Single Elimination Bracket~\cite{ssw17}, Randomized King of the Hill~\cite{swzz20}, Randomized Death Match~\cite{dw21}), which is optimal for $k=2$ among all Condorcet-consistent rules (that is, rules that select an undefeated team with probability $1$). 

Our main result establishes that Randomized Death Match is $3$-SNM-$(31/60)$, which is tight (for Randomized Death Match). This is the first tight analysis of any Condorcet-consistent tournament rule and at least three manipulating teams. Our proof approach is novel in this domain: we explicitly find the most-manipulable tournament, and directly show that no other tournament can be more manipulable.

In addition to our main result, we establish that Randomized Death Match disincentivizes Sybil attacks (where a team enters multiple copies of themselves into the tournament, and arbitrarily manipulates the outcomes of matches between their copies). Specifically, for any tournament, and any team $u$ that is not a Condorcet winner, the probability that $u$ or one of its Sybils wins in Randomized Death Match approaches $0$ as the number of Sybils approaches $\infty$.
\end{abstract}




\section{Introduction}
Consider a tournament on $n$ teams competing to win a single prize via $\binom{n}{2}$ pairwise matches. A tournament rule is a (possibly randomized) map  from these $\binom{n}{2}$ matches to a single winner. In line with several recent works~\cite{AK10,APT09,
ssw17,swzz20,dw21}, we study rules that satisfy some notion of fairness (that is, ``better'' teams should be more likely to win), and non-manipulability (that is, teams have no incentive to manipulate the matches).

More specifically, prior work identifies \textit{Condorcet-consistence} (Definition \ref{condorcetconsistent}) as one desirable property of tournament rules: whenever an undefeated team exists, a Condorcet-consistent rule selects that team as the winner with probability $1$. Another desirable property is \textit{monotonicity} (Definition \ref{monotonedef}): no team can unilaterally increase the probability that it wins by throwing a single match. Arguably, any sensible tournament rule should at minimum satisfy these two basic properties, and numerous such simple rules exist.

\cite{APT09,AK10} further considered the following type of deviation: what if the same company sponsors multiple teams in an eSports tournament, and wants to maximize the probability that one of them wins the top prize?\footnote{Similarly, perhaps there are multiple athletes representing the same country or university in a traditional sports tournament.} In principle, these teams might manipulate the outcomes of the matches they play amongst themselves in order to achieve this outcome. Specifically, they call a tournament rule $k$-Strongly-Non-Manipulable ($k$-SNM, Definition \ref{manipulatingatournament}), if no set of $k$ teams can successfully manipulate the $\binom{k}{2}$ pairwise matches amongst themselves to improve the probability that one of these $k$ teams wins the tournament. Unfortunately, even for $k=2$,~\cite{APT09, AK10} establish that no tournament rule is both Condorcet-consistent and $2$-SNM. 

This motivated recent work in~\cite{ssw17,swzz20,dw21} to design tournament rules which are Condorcet-consistent \emph{as non-manipulable as possible}. Specifically,~\cite{ssw17} defines a tournament rule to be $k$-SNM-$\alpha$ if no set of $k$ teams can manipulate the $\binom{k}{2}$ pairwise matches amongst themselves to increase total probability that any of these $k$ teams wins \emph{by more than $\alpha$} (see Definition~\ref{manipulatingatournament}). These works design several simple Condorcet-consistent and $2$-SNM-$1/3$ tournament rules, which is optimal for $k=2$~\cite{ssw17}. In fact, the state of affairs is now fairly advanced for $k=2$: each of~\cite{ssw17, swzz20, dw21} proposes a new $2$-SNM-$1/3$ tournament rule.~\cite{swzz20} considers a stronger fairness notion that they term Cover-consistent, and~\cite{dw21} considers probabilistic tournaments (see Section~\ref{sec:related} for further discussion).

However, \emph{significantly} less is known for $k > 2$. Indeed, only~\cite{swzz20} analyzes manipulability for $k > 2$. They design a rule that is $k$-SNM-$2/3$ for all $k$, but that rule is non-monotone, and it is unknown whether their analysis of that rule is tight. Our main result provides a tight analysis of the manipulability of Randomized Death Match~\cite{dw21} when $k=3$. We remark that this is: a) the first tight analysis of the manipulability of any Condorcet-consistent tournament rule for $k>2$, b) the first analysis establishing a monotone tournament rule that is $k$-SNM-$\alpha$ for any $k > 2$ and $\alpha < 1$, and c) the strongest analysis to-date of any tournament rule (monotone or not) for $k=3$. We overview our main result in more detail in Section~\ref{sec:results} below.\\
\indent Beyond our main result, we further consider manipulations through a \emph{Sybil attack} (Definition \ref{sybil_attack_def}). As a motivating example, imagine that a tournament rule is used as a proxy for a voting rule to select a proposal (voters compare each pair of proposals head-to-head, and this constitutes the pairwise matches input to a tournament rule). A proposer may attempt to manipulate the protocol with a Sybil attack, by submitting numerous nearly-identical clones of the same proposal. This manipulates the original tournament, with a single node $u_1$ corresponding to the proposal, into a new one with additional nodes $u_2,\ldots, u_m$ corresponding to the Sybils. Each node $v \notin \{u_1,\ldots, u_m\}$ either beats all the Sybils, or none of them (because the Sybil proposals are essentially identical to the original). The questions then become: Can the proposer profitably manipulate the matches within the Sybils? Is it beneficial for a proposer to submit as many Sybils as possible? We first show that, when participating in Randomized Death Match, the Sybils can't gain anything by manipulating the matches between them. Perhaps more surprisingly, we show that Randomized Death Match is \emph{Asymptotically Strongly Sybil-Proof}: as the number of Sybils approaches $\infty$, the collective probability that a Sybil wins RDM approaches \emph{zero} (unless the original proposal is a Condorcet winner, in which case the probability that a Sybil wins is equal to $1$, for any number of Sybils $> 0$).

\subsection{Our Results}\label{sec:results}

As previously noted, our main result is a tight analysis of the manipulability of Randomized Death Match (RDM) for coalitions of size $3$. Randomized Death Match is the following simple rule: pick two uniformly random teams who have not yet been eliminated, and eliminate the loser of their head-to-head match. 
\begin{inf_theorem} (See Theorem \ref{theorem2})
RDM is 3-SNM-$\frac{31}{60}$. RDM is not 3-SNM-$\alpha$ for $\alpha < \frac{31}{60}$.
\end{inf_theorem}

Recall that this is the first tight analysis of any Condorcet-consistent tournament rule for any $k>2$ and the first analysis establishing a monotone, Condorcet-consistent tournament rule that is $k$-SNM-$\alpha$ for any $k > 2$, $\alpha < 1$. Recall also that previously the smallest $\alpha$ for which a $3$-SNM-$\alpha$ (non-monotone) Condorcet-consistent tournament rule is known is $2/3$.

Our second result concerns manipulation by Sybil attacks. A Sybil attack is where one team starts from a base tournament $T$, and adds some number $m-1$ of clones of their team to create a new tournament $T'$ (they can arbitrarily control the matches within their Sybils, but each Sybil beats exactly the same set of teams as the cloned team) (See Definition \ref{sybil_attack_def}). We say that a tournament rule $r$ is \emph{Asymptotically Strongly Sybil-Proof} (Definition \ref{ASSP_def}) if for any tournament $T$ and team $u_1 \in T$ that is not a Condorcet winner, the maximum collective probability that a Sybil wins (under $r$) over all of $u_1$'s Sybil attacks with $m$ Sybils goes to 0 as $m$ goes to infinity. See Section \ref{sec:prelim} for a formal definition.


\begin{inf_theorem}(See Theorem~\ref{theoremcopiesto0}) RDM is Asymptotically Strongly Sybil-Proof.
\end{inf_theorem}


\subsection{Technical Highlight}
All prior work establishing that a particular tournament rule is $2$-SNM-$1/3$ follows a similar outline: for any $T$, cases where manipulating the $\{u,v\}$ match could potentially improve the chances of winning are coupled with two cases where manipulation cannot help. By using such a coupling argument, it is plausible that one can show that RDM is $3$-SNM-$(\frac{1}{2} + c)$ for a small constant $c$. However, given that Theorem~\ref{theorem2} establishes that RDM is $3$-SNM-$31/60$, it is hard to imagine that this coupling approach will be tractable to obtain the exact answer. 

Our approach is instead drastically different: we find a particular 5-team tournament, and a manipulation by $3$ teams that gains $31/60$, and directly prove that this must be the worst case. We implement our approach using a first-step analysis, thinking of the first match played in RDM on an $n$-team tournament as producing a distribution over $(n-1)$-team tournaments. 

The complete analysis inevitably requires some careful case analysis, but is tractable to execute fully by hand. Although this may no longer be the case for future work that considers larger $k$ or more sophisticated tournament rules, our approach will still be useful to limit the space of potential worst-case examples.


\subsection{Related Work}\label{sec:related}
There is a vast literature on tournament rules, both within Social Choice Theory, and within the broad CS community~\cite{Ban85,Fis77,FR92,LLB93,KSW16,KW15,Mil80,SW11}. The Handbook of Computational Social Choice provides an excellent survey of this broad field, which we cannot overview in its entirety~\cite{BCELP}. Our work considers the model initially posed in~\cite{AK10,APT09}, and continued in~\cite{dw21,ssw17,swzz20}, which we overview in more detail below.

\cite{AK10,APT09} were the first to consider 
Tournament rules that are both Condorcet-consistent and $2$-SNM, and proved that no such rules exist. They further considered tournament rules that are $2$-SNM and approximately Condorcet-consistent.~\cite{ssw17} first proposed to consider tournament rules that are instead Condorcet-consistent and approximately $2$-SNM. Their work establishes that Randomized Single Elimination Bracket is $2$-SNM-$1/3$, and that this is tight.\footnote{Randomized Single Elimination Bracket iteratively places the teams, randomly, into a single-elimination bracket, and then `plays' all matches that would occur in this bracket to determine a winner.}
\cite{swzz20} establish that Randomized King of the Hill (RKotH) is $2$-SNM-$1/3$,\footnote{Randomized King of the Hill iteratively picks a `prince', and eliminates all teams beaten by the prince, until only one team remains.} and~\cite{dw21} establish that Randomized Death Match is $2$-SNM-$1/3$.~\cite{swzz20} show further that RKotH satisfies a stronger fairness notion called Cover-consistence, and~\cite{dw21} extends their analysis to probabilistic tournaments. In summary, the state of affairs for $k=2$ is quite established: multiple $2$-SNM-$1/3$ tournament rules are known, and multiple different extensions beyond the initial model of~\cite{ssw17} are known.

For $k>2$, however, significantly less is known.~\cite{ssw17} gives a simple example establishing that no rule is $k$-SNM-$\frac{k-1-\varepsilon}{2k-1}$ for any $\varepsilon > 0$, but no rules are known to match this bound for any $k > 2$. Indeed,~\cite{swzz20} shows that this bound is not tight, and proves a stronger lower bound for $k \rightarrow \infty$. For example, a corollary of their main result is that no $939$-SNM-$1/2$ tournament rule exists. They also design a non-monotone tournament rule that is $k$-SNM-$2/3$ for all $k$. Other than these results, there is no prior work for manipulating sets of size $k > 2$. In comparison, our work is the first to give a tight analysis of any Condorcet-consistent tournament rule for $k > 2$, and is the first proof that any monotone, Condorcet-consistent tournament rule is $k$-SNM-$\alpha$ for any $k > 2,\alpha < 1$. \\
\indent Regarding our study of Sybil attacks, similar clone manipulations have been considered prior in Social Choice Theory under the name of \emph{composition-consistency}.  \cite{lll96} introduces the notion of a \emph{decomposition} of the teams in a tournament into components, where all the teams in a component are clones of each other with respect to the teams not in the component. \cite{lll96} defines a deterministic tournament rule to be \emph{composition-consistent} if it chooses the best teams from the best components\footnote{For a full rigorous mathematical definition see Definition 10, \cite{lll96}}. In particular, composition-consistency implies that a losing team cannot win by introducing clones of itself or any other team. \cite{lll96} shows that the tournament rules Banks, Uncovered Set, TEQ, and Minimal
Covering Set are \emph{composition-consistent}, while Top Cycle, the Slater, and the Copeland are not. Both computational and axiomatic aspects of \textit{composition-consistency} have been explored ever since. \cite{efs12} studies the structural properties of clone sets and their computational aspects in the context of voting preferences. In the context of probabilistic social choice, \cite{bbs16} gives probabilistic extensions of the axioms \emph{composition-consistency} and \emph{population-consistency} and uniquely characterize the probabilistic social choice rules, which satisfy both. In the context of scoring rules, \cite{o20} studies the incompatibility of \emph{composition-consistency} and \emph{reinforcement} (stronger than \emph{population-consistency}) and decomposes composition-consistency into four weaker axioms. In this work, we consider Sybil attacks on Randomized Death Match. Our study of Sybil attacks differs from prior work on the relevant notion of \emph{composition-consistency} in the following ways: (i) We focus on a randomized tournament rule (RDM), (ii) We study settings where the manipulator creates clones of themselves (i.e. not of other teams), (iii) We explore the asymptotic behavior of such manipulations (Definition \ref{ASSP_def}, Theorem \ref{theoremcopiesto0}).



\subsection{Roadmap}
Section \ref{sec:prelim} follows with definitions and preliminaries, and formally defines Randomized Death Match (RDM). Section \ref{examples} introduces some basic properties and examples for the RDM rule as well as a recap of previous work for two manipulators. Section \ref{rdm3} consists of a proof that the manipulability of 3 teams in RDM is at most $\frac{31}{60}$ and that this bound is tight. Section \ref{copiessection} consists of our main results regarding Sybil attacks on a tournament. Section ~\ref{sec:conclusion} concludes. 

\section{Preliminaries}\label{sec:prelim}
In this section we introduce notation that we will use throughout the paper consistent with prior work in \cite{dw21,ssw17,swzz20}.
\begin{definition}[Tournament]
A (round robin) tournament $T$ on $n$ teams is a complete,
directed graph on $n$ vertices whose edges denote the outcome of a match between two teams. Team $i$ beats
team $j$ if the edge between them points from $i$ to $j$.
\end{definition}

\begin{definition}[Tournament Rule]
A tournament rule $r$ is a function that maps tournaments $T$ to a distribution over teams, where $r_i(T) := \Pr(r(T) = i)$ denotes the probability that team $i$ is
declared the winner of tournament $T$ under rule $r$. We use the shorthand $r_S(T) := \sum_{i\in S} r_i(T)$ to denote
the probability that a team in $S$ is declared the winner of tournament $T$ under rule $r$.
\end{definition}

\begin{definition}[$S$-adjacent]\label{s_adjacent}
Two tournaments $T,T'$ are $S$-adjacent if for all $i,j$ such that $\{i,j\} \not \subseteq S$, $i$ beats $j$ in $T$ if and only if $i$ beats $j$ in $T'$. 
\end{definition}

In other words, two tournaments $T,T'$ are $S$-adjacent if the teams from $S$ can manipulate the outcomes of the matches between them in order to obtain a new tournament $T'$. 

\begin{definition}[Condorcet-Consistent]\label{condorcetconsistent}
Team $i$ is a Condorcet winner of a tournament $T$ if $i$ beats every
other team (under $T$). A tournament rule $r$ is Condorcet-consistent if for every tournament $T$ with a
Condorcet winner $i$, $r_i(T) = 1$ (whenever $T$ has a Condorcet winner, that team wins with probability 1).
\end{definition}

\begin{definition}[Manipulating a Tournament]\label{manipulatingatournament}
For a set $S$ of teams, a tournament $T$ and a tournament rule $r$, we define $\alpha_S^r(T)$ be the maximum winning probability that $S$ can possibly gain by manipulating $T$ to an $S$-adjacent $T'$. That is: 
$$\alpha_S^r(T) = \max_{\textit{T': T' is S-adjacent to T}} \{r_{S}(T') - r_{S}(T)\}$$
For a tournament rule $r$, define $\alpha_{k,n}^r = \sup_{T,S: |S|= k, |T| = n} \{\alpha_S^r(T)\}$. Finally, define 
$$\alpha_{k}^r = \sup_{n \in \mathbb{N}}\alpha_{k,n}^r =  \sup_{T,S: |S|= k} \{\alpha_S^r(T)\}$$
If $\alpha_k^r \leq \alpha$, we say that $r$ is $k$-Strongly-Non-Manipulable at probability $\alpha$ or $k$-SNM-$\alpha$. 
\end{definition}

Intuitively, $\alpha_{k,n}^r$ is the maximum increase in collective winning probability that a group of $k$ teams can achieve by manipulating the matches between them over tournaments with $n$ teams. And $\alpha_{k}^r$ is the maximum increases in winning probability that a group of $k$ teams can achieve by manipulating the matches between them over all tournaments. \\

Two other naturally desirable properties of a tournament rule are monotonicity and anonymity.
\begin{definition}[Monotone]\label{monotonedef}
A tournament rule $r$ is monotone if $T,T'$ are $\{u,v\}$-adjacent and $u$ beats $v$ in $T$, then $r_{u}(T) \geq r_{u}(T')$
\end{definition}

\begin{definition}[Anonymous]\label{anonimousdef}
A tournament rule $r$ is anonymous if for every tournament $T$, and every permutation $\sigma$, and all $i$, $r_{\sigma(i)}(\sigma(T)) = r_{i}(T)$
\end{definition}

Below we define the tournament rule that is the focus of this work.

\begin{definition}[Randomized Death Match]
Given a tournament $T$ on $n$ teams the Randomized Death Match Rule (RDM) picks two uniformly random teams (without replacement) and plays their match. Then, eliminates the loser and recurses on the remaining teams for a total of $n-1$ rounds until a single team remains, who is declared the winner. 
\end{definition}

Below we define the notions of \textit{Sybil Attack} on a tournament $T$, and the property of \textit{Asymptotically Strongly Sybil-Proof} (ASSP) for a tournament rule $r$, both of which will be relevant in our discussion in Section 5. 

\begin{definition}[Sybil Attack]\label{sybil_attack_def}
Given a tournament $T$, a team $u_1 \in T$ and an integer $m$, define $Syb(T,u_1,m)$ to be the set of tournaments $T'$ satisfying the following properties: \\
\indent 1. The set of teams in $T'$ consists of $u_2 \ldots, u_m$ and all teams in $T$\\
\indent 2. If $a,b$ are teams in $T$, then $a$ beats $b$ in $T'$ if and only if $a$ beats $b$ in $T$. \\
\indent 3. If $a \neq u_1$ is a team in $T$ and $i \in [m]$, then $u_i$ beats $a$ in $T'$ if and only if $u_1$ beats $a$ in $T$  \\
\indent 4. The match between $u_i$ and $u_j$ can be arbitrary for each $i \neq j$
\end{definition}

Intuitively, $Syb(T,u_1,m)$ is the set of all Sybil attacks of $u_1$ at $T$ with $m$ Sybils. Each Sybil attack is a tournament $T' \in Syb(T,u_1,m)$ obtained by starting from $T$ and creating $m$ Sybils of $u_1$ (while counting $u_1$ as a Sybil of itself). Each Sybil beats the same set of teams from $T \setminus u_1$ and the matches between the Sybils $u_1, \ldots, u_m$ can be arbitrary. Every possible realization of the matches between the Sybils gives rise to new tournament $T' \in Syb(T,u_1, m)$ (implying $Syb(T,u_1, m)$ contains $2^{\binom{m}{2}}$ tournaments). 

\begin{definition}[Asymptotically Strongly Sybil-Proof]\label{ASSP_def}
A tournament rule $r$ is \textit{Asymptotically Strongly Sybil-Proof} (ASSP) if for any tournament $T$, team $u_1 \in T$ which is not a Condorcet winner,
$$\lim_{m \to \infty}\max_{T' \in Syb(T,u_1,m)} r_{u_1, \ldots, u_m}(T') = 0$$
\end{definition}
Informally speaking, Definition \ref{ASSP_def} claims that $r$ is ASSP if the probability that a Sybil wins in the most profitable Sybil attack on $T$ with $m$ Sybils, goes to zero as $m$ goes to $\infty$.

\section{Basic Properties of RDM and Examples}\label{examples} 
In this section we consider a few basic properties of RDM and several examples on small tournaments. We will refer to those examples in our analysis later. Throughout the paper we will denote RDM by $r$ and it will be the only tournament rule we consider. We next state the first-step analysis observation that will be central to our analysis throughout the paper. For the remainder of the section let for a match $e$ denote by $T|_e$ the tournament obtained from $T$ by eliminating the loser in $e$. Let $S|_e = S \setminus x$, where $x$ is the loser in $e$. Let $d_x$ denote the number of teams $x$ loses to and $T \setminus x$ the tournament obtained after removing team $x$ from $T$. 
\begin{observation}[First-step analysis]\label{obs_fsa}
Let $S$ be a subset of teams in a tournament $T$. Then
$$r_{S}(T) = \frac{1}{\binom{n}{2}}\sum_{e}r_{S|_e}(T|_e) = \frac{1}{\binom{n}{2}}\sum_{x} d_x r_{S \setminus x}(T \setminus x)$$
(if $S = \{v\}$, then we define $r_{S \setminus v}(T \setminus v) = 0$, and if $x \not \in S$, then $S \setminus x = S$)
\end{observation}
\begin{proof}
The first equality follows from the fact that after we play $e$ we are left with the tournament $T|_e$ and we sum over all possible $e$ in the first round. To prove the second equality, notice that for any $x$ the term $r_{S \setminus x}(T \setminus x)$ appears exactly $d_x$ times in $\sum_{e} r_{S|_e}(T|_e)$ because $x$ loses exactly $d_x$ matches. 
\end{proof}\\

As a first illustration of first-step analysis we show that adding teams which lose to every other team does not change the probability distribution of the winner. 
\begin{lemma}\label{lemma1}
Let $T$ be a tournament and $u \in T$ loses to everyone. Then for all $v \neq u$, we have $r_v(T) = r_v(T \setminus u)$.
\end{lemma}
\begin{proof}
We prove the statement by induction on $|T|$. If $|T| = 2$, then clearly $r_v(T) = r_v(T \setminus u) = 1$. Suppose it holds on all tournaments $T'$ such that $|T'| < |T| = n$ and we will prove it for $T$. By 
first-step analysis (Observation \ref{obs_fsa}) we have that 
$$r_{v}(T) = \frac{1}{\binom{n}{2}}\sum_{e} r_{v}(T|_e) = \frac{1}{\binom{n}{2}} \sum_{x \neq v} d_x r_{v}(T \setminus x)$$
where team $x$ loses $d_x$ matches in $T$. By the inductive hypothesis we have that 
$r_{v}(T \setminus x) = r_{v}(T \setminus \{x,u\})$
for $x \neq u,v$ and $d_u = n-1$. Thus, 
\begin{align*}
    r_{v}(T) &= \frac{1}{\binom{n}{2}} \sum_{x \neq v} d_x r_{v}(T \setminus x) =\\
    &= \frac{1}{\binom{n}{2}}(\sum_{x \notin \{u,v\}} d_x r_{v}(T \setminus \{x,u\}) + (n-1)r_{v}(T \setminus u)) = \\
    &= \frac{1}{\binom{n}{2}}(\binom{n-1}{2}r_{v}(T \setminus u) + (n-1)r_{v}(T \setminus u)) =r_{v}(T \setminus u) \\
\end{align*}
where in the second to last line we used $\sum_{x \notin \{u,v\}} d_x r_{v}(T \setminus \{x,u\}) = \binom{n-1}{2}r_{v}(T \setminus u)$, which follows from first-step analysis (Observation \ref{obs_fsa}) and because $x$ loses $d_x$ matches in $T \setminus u$ as $u$ loses to every team in $T$
\end{proof}\\

As a natural consequence of Lemma \ref{lemma1} we show that the most manipulable tournament on $n+1$ teams is at least as manipulable as the most manipulable tournament on $n$ teams.
\begin{lemma}\label{lemma2}
$\alpha_{k,n}^r \leq \alpha_{k,n+1}^r$
\end{lemma}
\begin{proof} See Appendix~\ref{appendix_sec3} for a proof. 
\end{proof}\\

 We now show another natural property of RDM, which is a generalization of Condorcet-consistent (Definition \ref{condorcetconsistent}), namely that if a group of teams $S$ wins all its matches against the rest of teams, then a team from $S$ will always win. 
\begin{lemma}\label{lemma3}
Let $T$ be a tournament and $S \subseteq T$ a group of teams such that every team in $S$ beats every team in $T \setminus S$. Then, $r_{S}(T) = 1$.
\end{lemma}
\begin{proof} See Appendix~\ref{appendix_sec3} for a proof. 
\end{proof}\\

As a result of Lemma \ref{lemma3} RDM is Condorcet-Consistent. As expected, RDM is also monotone (See Definition \ref{monotonedef}). 
\begin{lemma}\label{rdm_monotone}
RDM is monotone.
\end{lemma}
\begin{proof} See Appendix~\ref{appendix_sec3} for a proof. 
\end{proof}\\

 Lemma \ref{lemma1} tells us that adding a team which loses to all other teams does not change the probability distribution of the other teams winning. Lemmas \ref{lemma1}, \ref{lemma2}, \ref{lemma3}, \ref{rdm_monotone} will be useful in our later analysis in Sections \ref{rdm3} and \ref{copiessection}. Now we consider a few examples of tournaments and illustrate the use of first-step analysis (Observation \ref{obs_fsa}) to compute the probability distribution of the winner in them.
\begin{enumerate}
    \item Let $T = \{a,b,c\}$, where $a$ beats $b$, $b$ beats $c$ and $c$ beats $a$. By symmetry of RDM, we have $r_a(T) = r_b(T) = r_c(T) = \frac{1}{3}$.
    \item Let $T = \{a,b,c\}$ where $a$ beats $b$ and $c$. Then clearly, $r_a(T) = 1$ and 
    $r_b(T) =r_c(T) = 0$.
    \item By Lemma \ref{lemma1}, it follows that the only tournament on 4 teams whose probability distribution cannot be reduced to a distribution on 3 teams is the following one $T = \{a_1,a_2,a_3,a_4\}$, where $a_i$ beats $a_{i+1}$ for $i = 1,2,3$, $a_4$ beats $a_1$, $a_1$ beats $a_3$ and $a_2$ beats $a_4$. By using what we computed in (1) and (2) combined with Lemma \ref{lemma1} we get by first step analysis
    \begin{align*}
    r_{a_1}(T) &= \frac{1}{6}(r_{a_1}(T \setminus a_2) +2 r_{a_1}(T \setminus a_3)+2 r_{a_1}(T \setminus a_4)) = \frac{1}{6}(\frac{1}{3} + \frac{2}{3} + 2) = \frac{1}{2}\\
    r_{a_2}(T) &= \frac{1}{6}(r_{a_2}(T \setminus a_1) +2 r_{a_2}(T \setminus a_3)+2 r_{a_2}(T \setminus a_4)) =\frac{1}{6}(1 +\frac{2}{3}) = \frac{5}{18}\\
    r_{a_3}(T) &= \frac{1}{6}(r_{a_3}(T \setminus a_1) +r_{a_3}(T \setminus a_2)+2 r_{a_3}(T \setminus a_4)) =\frac{1}{6}(\frac{1}{3}) = \frac{1}{18}\\
    r_{a_4}(T) &= \frac{1}{6}(r_{a_4}(T \setminus a_1) +r_{a_4}(T \setminus a_2)+2 r_{a_4}(T \setminus a_3)) =\frac{1}{6}(\frac{1}{3} + \frac{2}{3}) = \frac{1}{6}
    \end{align*}
\end{enumerate} 

The above examples are really important in our analysis because: a) we will use them in later for our lower bound example in Section \ref{lb}, and b) they are a short illustration of first-step analysis. \\
\indent In the following subsection, we review prior results for 2-team manipulations in RDM, which will also be useful for our treatment of the main result in Section \ref{rdm3}.

\subsection{Recap: Tight Bounds on 2-Team Manipulations in RDM}
\cite{dw21} (Theorem 5.2) proves that RDM is 2-SNM-$\frac{1}{3}$ and that this bound is tight, namely $ \alpha_2^{RDM} = \frac{1}{3}$. We will rely on this result in  Section \ref{rdm3}. 
\begin{theorem}\label{theorem1} (Theorem 5.2 in \cite{dw21})
$\alpha_2^{RDM} = \frac{1}{3}$
\end{theorem}

\cite{ssw17} (Theorem 3.1) proves that the bound of $\frac{1}{3}$ is the best one can hope to achieve for a Condorcet-consistent rule.

\begin{theorem}\label{theoremssw17} (Theorem 3.1 in \cite{ssw17})
There is no Condorcet-consistent tournament rule on $n$ players (for $n \geq 3$) that is 2-SNM-$\alpha$ for $\alpha < \frac{1}{3}$
\end{theorem}

We prove the following useful corollary, which will be useful in Section \ref{rdm3}.

\begin{corollary}\label{corollary}
Let $T$ be a tournament and $u,v \in T$ two teams such that there is at most one match in which a team in $\{u,v\}$ loses to a team in $T \setminus \{u,v\}$. Then
$$r_{u,v}(T) \geq \frac{2}{3}$$
\end{corollary}
\begin{proof}
If $u$ and $v$ beat every team in $T \setminus \{u,v\}$, then by Lemma \ref{lemma3}, $r_{u,v}(T) = 1 \geq \frac{2}{3}$. WLOG suppose that there is some team $t$ which beats $u$, loses to $v$ and all other teams lose to both $u$ and $v$. Let $T'$ be $\{u,v\}$-adjacent to $T$ such that $v$ is a Condorcet winner in $T'$. Clearly we have $r_{u,v}(T') = 1$ as RDM is Condorcet-Consistent. By Theorem \ref{theorem1} we have $r_{u,v}(T') -r_{u,v}(T) \leq \frac{1}{3}$. This, implies that $r_{u,v}(T) \geq \frac{2}{3}$ as desired.
\end{proof}

\section{Main Result: $\alpha^{RDM}_3 = 31/60$}\label{rdm3}
The goal of this section is to prove that no 3 teams can improve their probability of winning by more than $\frac{31}{60}$ and that this bound is tight. We prove the following theorem
\begin{theorem}\label{theorem2}
$\alpha_3^{RDM} = \frac{31}{60}$
\end{theorem}

Our proof consists of two parts: 
\begin{itemize}
    \item Lower bound: $\alpha_3^{RDM} \geq \frac{31}{60}$, for which we provide a tournament $T$ and a set $S$ of size 3, which can manipulate to increase their probability by $\frac{31}{60}$ 
    \item Upper bound: $\alpha_3^{RDM} \leq \frac{31}{60}$, for which we provide a proof that for any tournament $T$ no set $S$ of size 3 can increase their probability of winning by more than $\frac{31}{60}$, i.e. RDM is 3-SNM-$\frac{31}{60}$
\end{itemize}

\subsection{Lower Bound}\label{lb} Let $r$ denote RDM. Denote by $B_x$ the set of teams which team $x$ beats. Consider the following tournament $T = \{u,v,w,a,b\}$ (shown in Fig \ref{fig:example_graph}):
$$B_{a} = \{u,v,b\}, B_b = \{u,v\}, B_u = \{v,w\}, B_v = \{w\}, B_{w} = \{a,b\}$$
The tournament is also shown in Figure \ref{fig:example_graph}. 
Let $S = \{u,v,w\}$. By first-step analysis (Observation \ref{obs_fsa}) and by using our knowledge in Section \ref{examples} for tournaments on 4 teams we can write 
\begin{align*}
    r_{u,v,w}(T) &= \frac{1}{10}(3 r_{u,w}(T \setminus v) +2 r_{u,v}(T \setminus w) + 2 r_{u,v,w}(T \setminus b) + r_{u,v,w}(T \setminus a) + 2 r_{v,w}(T \setminus u)) = \\
    &= \frac{1}{10} (3\times(\frac{1}{2}+\frac{1}{6}) + 2 \times 0 + 2 \times (\frac{5}{18}+\frac{1}{18}+\frac{1}{6}) + (\frac{5}{18}+\frac{1}{18}+\frac{1}{6}) + 2 \times (\frac{1}{2}+\frac{1}{6}))= \\
    &= \frac{1}{10}(2+1 + \frac{1}{2} + \frac{4}{3}) = \frac{29}{60}
\end{align*}
Now suppose that $u$ and $v$ throw their matches with $w$. i.e. $T'$ is $S$-adjacent to $T$, where in $T'$, $w$ beats $u$ and $v$ and all other matches have the same outcomes as in $T$. Then, since $w$ is a Condorcet winner, $r_{u,v,w}(T') = r_{w}(T') = 1$. Therefore, 
$$\alpha_3^{RDM} \geq r_{u,v,w}(T')-r_{u,v,w}(T) = 1-\frac{29}{60} = \frac{31}{60}$$
Thus, $\alpha_3^{RDM} \geq \frac{31}{60}$ as desired. 

\begin{figure}[htp]
    \centering
    \includegraphics[width=6cm]{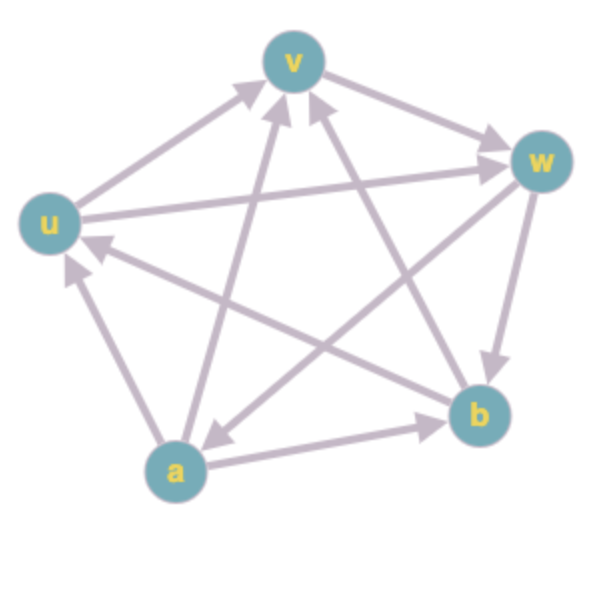}
    \caption{The tournament $T$ in which $S = \{u,v,w\}$ achieves a gain of $\frac{31}{60}$ by manipulation.}
    \label{fig:example_graph}
\end{figure}

\subsection{Upper Bound}
Suppose we have a tournament $T$ on $n \geq 3$ vertices and $S = \{u,v,w\}$ is a set of 3 (distinct) teams, where $S$ will be the set of manipulator teams. Let $I$ be the set of matches in which a team from $S$ loses to a team from $T \setminus S$. 
Our proof for $\alpha^{RDM}_k \leq \frac{31}{60}$ will use the following strategy 
\begin{itemize}
    \item In Sections \ref{FSA_framework} and \ref{bounds_on_abc} we introduce the first-step analysis framework by considering possible cases for the first played match. In each of these cases the loser of the match is eliminated and we are left with a tournament with one less team. We pair each match in $T$ with its corresponding match in $T'$ and we bound the gains of manipulation in each of the following cases separately (these correspond to each of the terms $A,B$, and $C$ respectively in the analysis in Section \ref{bounds_on_abc}).
        \begin{itemize}
            \item The first match is between two teams in $S$ (there are 3 such matches). 
            \item The first match is between a team in $S$ and a team in $T \setminus S$ and the team from $S$ loses in the match (there are $|I|$ such matches).
            \item The first match is any other match not covered by the above two cases
        \end{itemize}
    \item In Section \ref{caseI4} we prove that if $|I| \leq 4$, then $\alpha^{RDM}_{S}(T) \leq \frac{31}{60}$ (i.e. the set of manipulators cannot gain more than $\frac{31}{60}$ by manipulating). 
    \item In Section \ref{gen_upper_bound} we prove that if $T$ is the most manipulable tournament on $n$ vertices 
    (i.e. $\alpha^{RDM}_{S}(T) = \alpha_{3,n}^{RDM}$), then $\alpha^{RDM}_{S}(T) \leq \frac{|I| + 7}{3(|I| + 3)}$

    \item In Section \ref{final_proof} we combine the above facts to finish the proof of Theorem \ref{theorem2}
\end{itemize}

We first introduce some notation that we will use throughout this section. Suppose that $T'$ and $T$ are $S$-adjacent. Recall from Section \ref{examples} that for a match $e = (i,j)$, $T|_{e}$ is the tournament obtained after eliminating the loser in $e$. Also $d_x$ is the number of teams that a team $x$ loses to in $T$. For $x \in S$, let $\ell_x$ denote the number of matches $x$ loses against a team in $S$ when considered in $T$ and let $\ell'_x$ denote the number of matches that $x$ loses against a team in $S$ when considered in $T'$. Let $d^{*}_x$ denote the number of teams in $T \setminus S$ that $x$ loses to. Notice that since $T$ and $T'$ are $S$-adjacent, $x \in S$ loses to exactly $d^{*}_x$ teams in $T' \setminus S$ when considered in $T'$. Let $G = I \cup \{uv, vw, uw\}$ be the set of matches in which a team from $S$ loses.

\subsubsection{The First Step Analysis Framework}\label{FSA_framework}
Notice that in the first round of RDM, a uniformly random match $e$ from the $\binom{n}{2}$ matches is chosen. If $e \in G$ then we are left with $T \setminus x$ where $x$ loses in $e$ for some $x \in S$. If $e \not \in G$, we are left with $T|_{e}$ and all teams in $S$ are still in the tournament. For $x \in S$, there are $\ell_x$ matches in which they lose to a team from $S$ and $d^{*}_x$ matches in which they lose to a team from $T \setminus S$. By considering each of these cases and using first-step analysis (Observation \ref{obs_fsa}), we have
\begin{align*}
    r_{u,v,w}(T) &= \frac{1}{\binom{n}{2}} \Bigg[ \sum_{x \in \{u,v,w\}} \ell_x r_{\{u,v,w\} \setminus x}(T \setminus x) + d^{*}_u r_{v,w}(T \setminus u) + d^{*}_v r_{u,w}(T \setminus v) +d^{*}_w r_{v,u}(T \setminus w) \\&+ \sum_{e \notin G} r_{u,v,w}(T|_{e})  \Bigg]
\end{align*}
Since $T$ and $T'$ are $S$-adjacent each $x\in S$ loses to exactly $d^{*}_x$, teams form $T' \setminus S$, and we can similarly write
\begin{align*}
    r_{u,v,w}(T') &= \frac{1}{\binom{n}{2}}\Bigg[\sum_{x \in \{u,v,w\}} \ell'_x r_{\{u,v,w\} \setminus x}(T' \setminus x) + d^{*}_u r_{v,w}(T' \setminus u) + d^{*}_v r_{u,w}(T' \setminus v) + d^{*}_w r_{v,u}(T' \setminus w)\\
    &+\sum_{e \notin G} r_{u,v,w}(T'|_{e})\Bigg]\\
\end{align*}
By subtracting the above two expression we get 
\begin{align*}
    &r_{u,v,w}(T') - r_{u,v,w}(T) = \\
    &= \frac{1}{\binom{n}{2}} \Bigg[ \sum_{x \in \{u,v,w\}} \ell'_x r_{\{u,v,w\} \setminus x}(T' \setminus x) - \ell_x r_{\{u,v,w\}\setminus x}(T \setminus x)  + d^{*}_u(r_{\{v,w\}}(T' \setminus u)- r_{\{v,w\}}(T \setminus u)) \\
    &+ d^{*}_v(r_{\{u,w\}}(T' \setminus v)- r_{\{u,w\}}(T \setminus v)) +  d^{*}_w(r_{\{v,u\}}(T' \setminus w)- r_{\{v,u\}}(T \setminus w)) \\
    &+\sum_{e \notin G} r_{u,v,w}(T'|_{e})-r_{u,v,w}(T|_{e}) \Bigg]
\end{align*}
Thus,
\begin{equation}\label{fsa:1}
    r_{u,v,w}(T') - r_{u,v,w}(T) = \frac{1}{\binom{n}{2}}(A+B+C)
\end{equation}
where 
\begin{align*}
A &= \sum_{x \in \{u,v,w\}} \ell'_x r_{\{u,v,w\} \setminus x}(T' \setminus x) - \ell_x r_{\{u,v,w\}\setminus x}(T \setminus x)\\
B &= \sum_{x \in S} d^{*}_x (r_{\{u,v,w\} \setminus x}(T' \setminus x)- r_{\{u,v,w\} \setminus x}(T \setminus x))\\
C &= \sum_{e \notin G} r_{u,v,w}(T'|_{e})-r_{u,v,w}(T|_{e})
\end{align*}

\subsubsection{Upper Bounds on $A,B$ and $C$}\label{bounds_on_abc}
We now prove some bounds on the terms $A$, $B$ and $C$ (defined in Section \ref{FSA_framework}) which will be useful later. Recall that $I$ denotes the set of matches in which a team from $S$ loses from a team from $T \setminus S$. We begin with bounding  $A$ in the following lemma

\begin{lemma}\label{lemmaA}
For all $S$-adjacent $T$ and $T'$, we have $A \leq \frac{7}{3}$. Moreover, if $|I| \leq 1$, then $A \leq 1$.
\end{lemma}
\begin{proof} See Appendix \ref{appendix_sec4} for a proof. 
\end{proof}\\

Next, we show the following bound on the term $B$. 
\begin{lemma}\label{lemmaB}
For all $S$-adjacent $T,T'$  we have 
$$B \leq \frac{d^{*}_u + d^{*}_v + d^{*}_w}{3} = \frac{|I|}{3}$$
Moreover, if $|I| \leq 1$, then $B = 0$
\end{lemma}
\begin{proof}See Appendix \ref{appendix_B} for a proof. 
\end{proof}\\

We introduce some more notation. For $n \in \mathbb{N}$, define $M_n(a_1,a_2,a_3)$ as the maximum winning probability gain that three teams $\{u,v,w\}$ can achieve by manipulation in a tournament $T$ of size $n$, in which there are exactly $a_i$ teams in $T \setminus S$ each of which beats exactly $i$ teams of $S$.
Formally,
\begin{align*}
M_n(a_1,a_2,a_3) &= \max \Big\{r_{S}(T')-r_{S}(T) | T,T' \text{ are } S\text{-adjacent},|T|=n, |S|=3, \\
&\text{ $a_i$ teams in $T \setminus S$ beat exactly $i$ teams in $S$} \Big\}
\end{align*}
\indent Additionally, let $L_i$ be the set of teams in $T \setminus S$ each of which beats exactly $i$ teams in $S$. Let $Q$ be the set of matches in which two teams from $L_i$ play against each other or in which a team from $L_i$ loses to a team from $S$ for $i = 1,2,3$. Notice that $|Q| = 2a_1 + a_2 + \binom{a_1}{2} + \binom{a_2}{2} + \binom{a_3}{2}$ if there are $a_i$ teams in $S \setminus T$ each which beat $i$ teams from $S$. \\
\indent  With the new notation, we are now ready to prove a bound on the term $C$. 
Recall that 
$$C = \sum_{e \notin G} r_{u,v,w}(T'|_{e})-r_{u,v,w}(T|_{e})$$
where $G = I \cup \{uv, vw, uw\}$ is the set of matches in which a team from $S$ loses. Then we have the following bound on $C$. 
\begin{lemma}\label{lemmaC}
For all $S$-adjacent $T$ and $T'$ we have that $C$ is at most
\begin{align*}
&(2a_1 + \binom{a_1}{2})M_{n-1}(a_1-1,a_2,a_3)+ (a_2 + \binom{a_2}{2}) M_{n-1}(a_1,a_2-1,a_3) +  \binom{a_3}{2}M_n(a_1,a_2,a_3-1)\\ &+\sum_{e \notin G \cup Q} r_{u,v,w}(T'|_{e})-r_{u,v,w}(T|_{e})
     \end{align*}
\end{lemma}
\begin{proof} See Appendix \ref{appendix_C} for a proof. 
\end{proof}

\subsubsection{ The case $|I| \leq 4$}\label{caseI4} We summarize our claim when $|I| \leq 4$ in the following lemma

\begin{lemma}\label{lemmaI4}
Let $T$ be a tournament, and $S$ a set of 3 teams. Suppose that there are at most 4 matches in which a team in $S$ loses to a team in $T \setminus S$ (i.e. $|I| \leq 4$). Then 
$\alpha^{RDM}_{S}(T) \leq \frac{31}{60}$
\end{lemma}
\begin{proof}
We will show that $M_n(a_1,a_2,a_3) \leq f(a_1,a_2,a_3)$ by induction on $n \in \mathbb{N}$ for the values of $(a_1,a_2,a_3)$ and $f(a_1,a_2,a_3)$ given in Table \ref{tab:upp_bounds} below. Notice that when there are at most 4 matches between a team in $S$ and a team in $T \setminus S$, in which the teams from $S$ loses, then we fall into one of the cases shown in the table for $(a_1,a_2,a_3)$.

\setlength{\tabcolsep}{18pt} 
\renewcommand{\arraystretch}{1.6} 
\begin{table}[h]\label{table}
    \scriptsize
    \centering
    \begin{tabular}{|c|c|}
        \hline
$(a_1,a_2,a_3)$&$f(a_1,a_2,a_3)$\\
\hline
(0,0,0)& 0\\
\hline
(1,0,0)&$\frac{1}{6}$  \\
\hline
(2,0,0) & $\frac{23}{60}$ \\
\hline
(3,0,0) & $\frac{407}{900}$ \\
\hline
(4,0,0) & $\frac{4499}{9450}$  \\
\hline
(0,1,0)& $\frac{1}{2}$\\
\hline
(0,2,0) & $\frac{31}{60}$ \\
\hline
(1,1,0)& $\frac{1}{2}$\\
\hline
(2,1,0)& $\frac{131}{260}$ \\
\hline
(0,0,1)&0 \\
\hline
(1,0,1)&$\frac{11}{27}$\\
\hline
    \end{tabular}
     \caption{Upper bounds on $M_n(a_1,a_2,a_3)$}
    \label{tab:upp_bounds}
\end{table}

\indent \textbf{1. Base case} Our base case is when $n = 3$. If we are in the case of 3 teams then $S$ wins with probability 1, so the maximum gain $S$ can achieve by manipulation is clearly 0, which satisfies all of the bounds in the table. \\
\indent \textbf{2. Induction step} Assume that $M_k(a_1,a_2,a_3) \leq f(a_1,a_2,a_3)$ hold for all $ k < n$ and $a_1,a_2,a_3$ as in Table \ref{tab:upp_bounds}. We will prove the statement for $k = n$. Notice that by Table \ref{tab:upp_bounds} it is clear that $f$ is monotone in each variable i.e. if $a'_i \leq a_i$ for $i = 1,2,3$, then 
\begin{equation}\label{eq:monotonicity}
    f(a'_1,a'_2,a'_3) \leq f(a_1,a_2,a_3)
\end{equation}

Suppose that $e \notin G \cup Q$\footnote{Recall $G$ is the set of matches in which a team from $S$ loses and $Q$ is the set of matches in between two teams from $L_i$ or when a team from $L_i$ loses to a team from $S$ (see discussion before Lemma \ref{lemmaC})}. Then since $e \notin G$, $S$ remains in $T|_{e}$. Let for a tournament $H$ define $t(H) = (a'_1,a'_2,a'_3)$, where in $H$ there are exactly $a'_i$ teams in $H \setminus \{u,v,w\}$ that beat exactly $i$ out of $\{u,v,w\}$. Clearly, we have $t(T|_e) = (a'_1, a'_2,a'_3)$, where $a'_i \leq a_i$. By monotonicity of $f$ in (\ref{eq:monotonicity}) and the inductive hypothesis it follows that 
$$r_{u,v,w}(T'|_{e})-r_{u,v,w}(T|_{e}) \leq M_{n-1}(a'_1,a'_2,a'_3) \leq f(a'_1,a'_2,a'_3) \leq f(a_1,a_2,a_3)$$
Since $|G \cup Q| = 3(1+a_1+a_2+a_3) + \binom{a_1}{2} + \binom{a_2}{2} +\binom{a_3}{2}$, we have that
$$\sum_{e \notin G \cup Q} r_{u,v,w}(T'|_{e})-r_{u,v,w}(T|_{e}) \leq (\binom{n}{2}-3(1+a_1+a_2+a_3)-\binom{a_1}{2}-\binom{a_2}{2}-\binom{a_3}{2})f(a_1,a_2,a_3)$$
Also, by the inductive hypothesis we have 
\begin{align*}
    M_{n-1}(a_1-1,a_2,a_3) &\leq f(a_1-1,a_2, a_3)\\
    M_{n-1}(a_1,a_2-1,a_3) &\leq f(a_1,a_2-1, a_3)\\
    M_{n-1}(a_1,a_2,a_3-1) &\leq f(a_1,a_2, a_3-1)\\
\end{align*}
Therefore, by Lemma \ref{lemmaC}, combined with the inequalities above we have 
\begin{align*}
     C &\leq (2a_1 + \binom{a_1}{2})f(a_1-1,a_2,a_3)+ (a_2 + \binom{a_2}{2}) f(a_1,a_2-1,a_3) + \binom{a_3}{2}f(a_1,a_2,a_3-1) \\
     &+(\binom{n}{2}-3(1+a_1+a_2+a_3)-\binom{a_1}{2}-\binom{a_2}{2}-\binom{a_3}{2})f(a_1,a_2,a_3)
\end{align*}
By Lemma \ref{lemmaB} we have 
$$B \leq \frac{d^{*}_u + d^{*}_v + d^{*}_w}{3} = \frac{a_1 + 2a_2 + 3a_3}{3} = \frac{|I|}{3}$$
Combining the above with bounds and plugging into $(\ref{fsa:1})$ we get
\begin{align*}
        &r_{u,v,w}(T') - r_{u,v,w}(T) \leq \frac{1}{\binom{n}{2}}(A+B+C) \leq \\
        &\leq \frac{1}{\binom{n}{2}}\Bigg[A' + B'+(2a_1 + \binom{a_1}{2})f(a_1-1,a_2,a_3)+ (a_2 + \binom{a_2}{2}) f(a_1,a_2-1,a_3) \\ &+\binom{a_3}{2}f(a_1,a_2,a_3-1) + (\binom{n}{2}-3(1+a_1+a_2+a_3)-\binom{a_1}{2}-\binom{a_2}{2}-\binom{a_3}{2})f(a_1,a_2,a_3)\Bigg]
\end{align*}
where $A' = 1$ and $B' = 0$ if $|I| \leq 1$ and $A' = \frac{7}{3}$ and $B' = \frac{|I|}{3}$ if $|I| \geq 2$ by Lemma \ref{lemmaA} and Lemma \ref{lemmaB}. As the RHS depends only on $(a_1,a_2,a_3)$ we can take the maximum over all tournaments on $n$ teams so we can get 
\begin{align*}
    M_n(a_1,a_2,a_3) &\leq \frac{1}{\binom{n}{2}}\Bigg[A' + B' +(2a_1 + \binom{a_1}{2})f(a_1-1,a_2,a_3)+ (a_2 + \binom{a_2}{2}) f(a_1,a_2-1,a_3) \\ &+\binom{a_3}{2}f(a_1,a_2,a_3-1) \\
    &+ (\binom{n}{2}-3(1+a_1+a_2+a_3)-\binom{a_1}{2}-\binom{a_2}{2}-\binom{a_3}{2})f(a_1,a_2,a_3)\Bigg] \text{   ($\Delta$)}
\end{align*}
Now, apply the formula $(\Delta)$ to each of the cases in Table  \ref{tab:upp_bounds}. We present the computations for $(a_1, a_2, a_3) \in \{(0,0,0), (1,0,0), (0,2,0)\}$ in the body to illustrate the method and we defer the other cases from Table \ref{tab:upp_bounds} to Appendix \ref{appendix_sec4}. Note the manipulators can achieve $\frac{31}{60}$ only when $(a_1,a_2,a_3) = (0,2,0)$.

\textbf{Case 1} $(a_1,a_2,a_3) = (0,0,0)$.  By Lemma \ref{lemma3} it follows that $M_n(0,0,0) = 0 = f(0,0,0)$ as a team from $S$ wins with probability 1 regardless of the matches within $S$. \\

\textbf{Case 2} $(a_1,a_2,a_3) = (1,0,0)$. In this case $|I| = 1$. So by applying $\Delta$ with $A' = 1$ and $B = 0$ we obtain 
$$M_n(1,0,0) \leq \frac{1}{\binom{n}{2}}(1 + 2 f(0,0,0) + (\binom{n}{2}-6)\frac{1}{6}) = \frac{1}{6} = f(1,0,0)$$
 \\

\textbf{Case 3} $(a_1,a_2,a_3) = (0,2,0)$. Applying $\Delta$, we get
\begin{align*}
M_n(0,2,0) &\leq \frac{1}{\binom{n}{2}}(\frac{7}{3} + \frac{4}{3} +  (2 + 1)f(0,1,0) + (\binom{n}{2}-10)f(0,2,0)) \\
    &=  \frac{1}{\binom{n}{2}}(\frac{11}{3}  +  \frac{3}{2} + (\binom{n}{2}-10)\frac{31}{60}) \\
    &= \frac{31}{60} + \frac{1}{\binom{n}{2}}(\frac{22+9}{6} - \frac{31}{6}) = \frac{31}{60} = f(0,2,0)
\end{align*}

\textbf{Case 4} $(a_1,a_2,a_3) = (2,0,0)$. See Appendix ~\ref{appendix_cases}. 


\textbf{Case 5} $(a_1,a_2,a_3) = (3,0,0)$. See Appendix ~\ref{appendix_cases}. 


\textbf{Case 6} $(a_1,a_2,a_3) = (4,0,0)$. See Appendix ~\ref{appendix_cases}.


\textbf{Case 7} $(a_1,a_2,a_3) = (0,1,0)$. See Appendix ~\ref{appendix_cases}. 


\textbf{Case 8} $(a_1,a_2,a_3) = (1,1,0)$. See Appendix ~\ref{appendix_cases}.


\textbf{Case 9} $(a_1,a_2,a_3) = (2,1,0)$. See Appendix ~\ref{appendix_cases}.


\textbf{Case 10} $(a_1,a_2,a_3) = (0,0,1)$. See Appendix ~\ref{appendix_cases}. 


\textbf{Case 11} $(a_1,a_2,a_3) = (1,0,1)$. See Appendix ~\ref{appendix_cases}. 


This, finishes the induction and the proof for the bounds in Table \ref{tab:upp_bounds}. Note that $f(a_1,a_2,a_3) \leq \frac{31}{60}$ for all $a_1,a_2,a_3$ in Table \ref{tab:upp_bounds} and this bounds is achieved when $(a_1,a_2,a_3) = (0,2,0)$ i.e. there are 2 teams that beat exactly two of $S$ as is the case in the optimal example in Section \ref{lb}. Thus,we get that if there are at most 4 matches that a team from $S$ loses from a team in $T \setminus S$, then $\alpha^{RDM}_S(T) \leq \frac{31}{60}$. This finishes the proof of the lemma. 
\end{proof}

\subsubsection{General Upper Bound for the Most Manipulable Tournament}\label{gen_upper_bound}

\begin{lemma}\label{lemmaI5}
Suppose that $\alpha^{RDM}_{S}(T) = \alpha^{RDM}_{3,n}$. Let $I$ be the set of matches a team of $S$ loses to a team from $T \setminus S$. Then $$\alpha^{RDM}_{3,n}= \alpha^{RDM}_S(T) \leq \frac{|I|+7}{3(|I|+3)}$$
\end{lemma}
\begin{proof}
Let $T$ and $T'$ be $S$-adjacent tournaments on $n$ vertices such that $S = \{u,v,w\}$ and 
$$\alpha_{3,n}^{RDM} = \alpha^{RDM}_S(T) = r_S(T') - r_s(T)$$
I.e. $T$ is the "worst" example on $n$ vertices. By (\ref{fsa:1}) we have 
$$
    \alpha_{3,n}^{RDM} = \frac{1}{\binom{n}{2}}(A+B+C)
$$
where $A,B$ and $C$ were defined in Section \ref{FSA_framework}. By Lemma \ref{lemmaA} we have
$$A \leq \frac{7}{3}$$
and by Lemma \ref{lemmaB}
$$B \leq \frac{d^{*}_u + d^{*}_v + d^{*}_w}{3} = \frac{|I|}{3}$$

Let $e \notin G$. Notice that both $T'|_{e}$ and $T|_{e}$ are tournaments on $n-1$ vertices and by definition of $G$, $u,v,w$ are not eliminated in both $T'|_{e}$ and $T|_{e}$. Moreover, $T'|_{e}$ and $T|_{e}$ are $S$-adjacent. Therefore, for every $e \notin G$, we have by Lemma \ref{lemma2}
$$
r_{u,v,w}(T'|_{e})-r_{u,v,w}(T|_{e}) \leq \alpha_{3,n-1}^{RDM} \leq \alpha_{3,n}^{RDM}
$$
By using the above on each term in $C$ and the fact that $|G| = 3 +|I|$, we get that 
$$
C \leq (\binom{n}{2} -(3 + |I|))\alpha_{3,n}^{RDM}
$$
By using the above 3 bounds we get 
\begin{align*}
 \alpha_{3,n}^{RDM} &\leq \frac{1}{\binom{n}{2}} (\frac{7}{3}+\frac{|I|}{3} + (\binom{n}{2} -(3 + |I|))\alpha_{3,n}^{RDM}) \\
    \iff (|I|+3)\alpha_{3,n}^{RDM} &\leq \frac{|I|+7}{3}\\
    \iff \alpha_{3,n}^{RDM} = \alpha^{RDM}_S(T) &\leq \frac{|I|+7}{3(|I|+3)}
\end{align*}
as desired.
\end{proof}

\subsubsection{Proof of Theorem \ref{theorem2}}\label{final_proof} Suppose that $T$ is the most manipulable tournament on $n$ vertices i.e. it satisfies $\alpha^{RDM}_{S}(T) = \alpha^{RDM}_{3,n}$. If $|I| \leq 4$, then by Lemma \ref{lemmaI4}, we have that 
$$\alpha_{3,n}^{RDM} = \alpha^{RDM}_S(T) \leq \frac{31}{60}$$
If $|I| \geq 5$, then by Lemma \ref{lemmaI5}
$$\alpha_{3,n}^{RDM} = \alpha^{RDM}_S(T) \leq \frac{|I|+7}{3(|I|+3)} \leq \frac{5+7}{3(5+3)} = \frac{1}{2}$$
where above we used that $\frac{x+7}{3(x+3)}$ is decreasing for $x \geq 5$. Combining the above bounds, we obtain $\alpha_{3,n}^{RDM} \leq \frac{31}{60}$ for all $n \in \mathbb{N}$. Therefore, 
$$\alpha_{k}^{RDM} = \max_{n \in \mathbb{N}}\alpha_{k,n}^{RDM} \leq \frac{31}{60}$$
which proves the upper bound and finishes the proof of Theorem \ref{theorem2}.

\section{Sybil Attacks on Tournaments}\label{copiessection}

\subsection{Main Results on Sybil Attacks on Tournaments}

Recall our motivation from the Introduction. Imagine that a tournament rule is used as a proxy for a voting rule to select a proposal. The proposals are compared head-to-head, and this constitutes the pairwise matches in the resulting tournament. A proposer can try to manipulate the protocol with a Sybil attack and submit many nearly identical proposals with nearly equal strength relative to the other proposals. The proposer can choose to manipulate the outcomes of the head-to-head comparisons between two of his proposals in a way which maximizes the probability that a proposal of his is selected. In the tournament $T$ his proposal corresponds to a team $u_1$, and the tournament $T'$ resulting from the Sybil attack is a member of $Syb(T,u_1,m)$ (Recall Definition 2.9). The questions that we want to answer in this section are: (1) Can the Sybils manipulate their matches to successfully increase their collective probability of winning? and (2) Is it beneficial for the proposer to create as many Sybils as possible?

The first question we are interested is whether any group of Sybils can manipulate successfully to increase their probability of winning. It turns out that that the answer is No. We first prove that the probability that a team which is not a Sybil wins does not depend on the matches between the Sybils.

\begin{lemma}\label{lemmacopies}
There exists a function $q$ that takes in as input integer $m$, tournament $T$, team $u_1 \in T$, and team $v \in T \setminus u_1$ with the following property. For all $T' \in Syb(T,u_1,m)$, we have
$$r_{v}(T')  = q(m,T,u_1,v)$$
where the dependence on $u_1$ is encoded as the outcomes of its matches with the rest of $T$. 
\end{lemma}
\begin{proof}
See Appendix ~\ref{appendix_sec5} for a full proof. 
\end{proof}\\

Note that by Lemma \ref{lemmacopies} $r_v(T') = q(m,T,u_1,v)$ does not depend on which tournament $T' \in Syb(T, u_1,m)$ is chosen. Now, we prove our first promised result. Namely, that no number of Sybils in a Sybil attack can manipulate the matches between them to increase their probability of winning.

\begin{theorem}\label{theorem4}
Let $T$ be a tournament, $u_1 \in T$ a team, and $m$ and integer. Let $T_1' \in Syb(T,u_1,m)$. Let $S = \{u_1, \ldots, u_m\}$. Then 
$$\alpha_{S}^{RDM}(T_1') = 0$$

\end{theorem}
\begin{proof}
Let $T_1'$ and $T_2'$ be $S$-adjacent. By Definition 2.9, $T_2' \in Syb(T,u_1,m)$. Therefore by Lemma \ref{lemmacopies}, $r_v(T_1') = r_v(T_2') = q(m,T,u_1,v)$ for all $v \in T \setminus u_1$. Using this we obtain
$$r_{S}(T_1') =1 - \sum_{v \in T \setminus u_1} r_{v}(T_1') = 1 - \sum_{v \in T \setminus u_1} r_{v}(T_2') = r_{S}(T_2')$$
Therefore, $r_{S}(T_1') = r_{S}(T_2')$ for all $S$-adjacent $T_1', T_2'$, which implies the desired result. 
\end{proof}\\

Theorem \ref{theorem4} says that any number of Sybils cannot manipulate to increase their collective probability of winning. This leaves the question whether it is beneficial for the proposer to send many (nearly) identical proposals to the tournament to maximize the probability that a proposal of his is selected. We show that Randomized Death Match disincentivizes such behaviour. When $u_1$ is a Condorcet winner in $T$, then by Lemma \ref{lemma3} a Sybil will win with probability one. 
We prove that if $u_1$ is not a Condorcet winner in $T$ then as $m$ goes to infinity, the maximum probability (over all Sybils attacks of $u_1$) that any Sybil wins goes to 0. Or equivalently, RDM is Asymptotoically Strongly Sybil-Proof (recall Definition 2.10).\\
\indent Before we state our second main theorem let's recall some notation. Let $u_1 \in T$ be a team (which is not a Condorcet winner). Let $A$ be the set of teams in $T$ that $u_1$ beats, and $B$ the of teams $u_1$ loses to. Let $T' \in Syb(T,u_1, m)$ and $v \in T \setminus u_1$. By Lemma \ref{lemmacopies}, $r_v(T') = q(m, T, u_1,v)$ for all $T' \in Syb(T,u_1,m)$. Also, by Lemma \ref{lemmacopies}, 
\begin{equation}\label{prob_sybils}
r_{u_1, \ldots, u_m}(T') = 1-\sum_{v \in T \setminus u_1} r_v(T') = 1-\sum_{v \in T \setminus u_1} q(m,T,u_1,v) 
\end{equation}
and 
\begin{equation}\label{prob_A}
r_{A}(T') = \sum_{v \in A}q(m,T,u_1,v)
\end{equation}
Note that the terms in the RHS of each of (\ref{prob_sybils}) and (\ref{prob_A}) depends only on $T, u_1$ and $m$. Thus, we can define the functions
$$h(m,T,u_1) = r_{u_1, \ldots, u_m}(T')$$
$$g(m,T,u_1) = r_{A}(T')$$
and 
$$p(m,T,u_1) = h(m,T,u_1) + g(m,T,u_1)$$
(here $p(m,T,u_1)$ is the total collective probability that a Sybil or a team from $A$ wins) \\
\indent We are now ready to prove our second main result. Namely, that RDM is Asymptotically Strongly Sybil-Proof (Definition \ref{ASSP_def}). Before we present the result (Theorem \ref{theoremcopiesto0}) we will try to convey some intuition for why RDM should be ASSP. Observe that the only way a Sybil can win is when all the teams from $B$ are eliminated before all the Sybils. The teams from $B$ can only be eliminated by teams from $A$. However, as $m$ increases there are more Sybils and, thus, the teams from $A$ are intuitively more likely to all lose the tournament before the teams from $B$. When there are no teams from $A$ left and at least one team from $B$ left, no Sybil can win. In fact, this intuition implies something stronger than RDM being ASSP: the collective winning probability of the Sybils and the teams from $A$ (denoted by $p(m,T,u_1)$) converges to 0 as $m$ converges to $\infty$ (or, equivalently, the probability that a team from $B$ wins goes to $1$). This intuition indirectly lies behind the technical details of the proof of Theorem \ref{theoremcopiesto0}.

\begin{theorem}\label{theoremcopiesto0}
Randomized Death Match is Asymptotically Strongly Sybil-Proof. In fact a stronger statement holds, namely if $u_1 \in T$ is not a Condorcet winner, then
$$\lim_{m \to \infty} p(m,T,u_1) = 0$$
\end{theorem}
\begin{proof}See Appendix ~\ref{appendix_sec5} for a full proof. 
\end{proof} \\


\subsection{On a Counterexample to an Intuitive Claim}
We will use Theorem \ref{theoremcopiesto0} to prove that RDM does not satisfy a stronger version of the monotonicity property in Definition \ref{monotonedef}. First, we give a generalization of the definition for monotonicity given in Section \ref{examples}
\begin{definition}[Strongly monotone]
Let $r$ be a tournament rule. Let $T$ and let $C \cup D$ be any splitting of the teams in $T$ into two disjoint sets. A tournament rule $r$ is \textit{strongly monotone} for every $(u,v) \in C \times D$, if $T'$ is $\{u,v\}$-adjacent to $T$ such that $u$ beats $v$ in $T'$ we have $r_{C}(T') \geq r_{C}(T)$
\end{definition}
Intuitively, $r$ is Strongly monotone if whenever flipping a match between a team from $C$ and a team from $D$ in favor of the team from $C$ makes $C$ better off. Notice that if $|C| = 1$ this is the usual definition of monotonicity (Definition \ref{monotonedef}), which is satisfied by RDM by Lemma \ref{rdm_monotone}. However, RDM is not strongly monotone, even though strong monotonicity may seem like an intuitive property to have. 
\begin{claim}\label{rdm_not_strongly_monotne}
\label{claim_rdm_not_monotone}
RDM is not strongly monotone
\end{claim}
\begin{proof} 
Suppose the contrary i.e. RDM is strongly monotone. Start with a 3-cycle tournament $T$ where $a_1$ beats $b$, $b$ beats $c$, and $c$ beats $a_1$. Let $T' \in Syb(T,a_1,m)$ be a Sybil attack of $a_1$ on $T$ with $m$ Sybils. Let the Sybils be $C = \{a_1,a_2, \ldots, a_m\}$ where $a_i$ beats $a_j$ in $T'$ for $i < j$.  By Theorem \ref{theoremcopiesto0} we can take $m$ large enough so that $r_{C}(T') < \frac{1}{6}$. Suppose all Sybils in $C$ but $a_1$ throw all of their matches with $b$  and denote the tournament resulting from that $T''$. Then, if RDM were strongly monotone we would have $r_{C}(T'') \leq r_{C}(T') < \frac{1}{6}$ (start with $T''$ and iteratively apply strong monotonicity). Note that starting from $T''$ and iteratively applying Lemma \ref{lemma1} to $a_{i}$ and removing $a_i$ from the tournament for $i = m,m-1, \ldots, 2$ we will obtain the tournament $T$, and the probability distribution over the winner in $T''$ will be the same as in $T$. Therefore, $r_{C}(T'') \geq r_{a}(T'') = r_{a}(T) = \frac{1}{3}$, but $r_{a}(T'') < \frac{1}{6}$, a contradiction with our assumption that RDM is strongly monotone. 
\end{proof}

\section{Conclusion and Future Work}\label{sec:conclusion}
We use a novel first-step analysis to nail down the minimal $\alpha$ such that RDM is $3$-SNM-$\alpha$. Specifically, our main result shows that $\alpha^{RDM}_3 = \frac{31}{60}$. Recall that this is the first tight analysis of any Condorcet-consistent tournament rule for any $k > 2$, and also the first monotone, Condorcet-consistent tournament rule that is $k$-SNM-$\alpha$ for any $k > 2,\alpha < 1$. We also initiate the study of manipulability via Sybil attacks, and prove that RDM is Asymptotically Strongly Sybil-Proof. 

Our technical approach opens up the possibility of analyzing the manipulability of RDM (or other tournament rules) whose worst-case examples are complicated-but-tractable. For example, it is unlikely that the elegant coupling arguments that work for $k=2$ will result in a tight bound of $31/60$, but our approach is able to drastically reduce the search space for a worst-case example, and a tractable case analysis confirms that a specific 5-team tournament is tight. Our approach can similarly be used to analyze the manipulability of RDM for $k > 3$, or other tournament rules. However, there are still significant technical barriers for future work to overcome in order to keep analysis tractable for large $k$, or for tournament rules with a more complicated recursive step. Still, our techniques provide a clear approach to such analyses that was previously non-existent.

\bibliographystyle{alpha}
\bibliography{References} 

\appendix
\section{Omitted proofs}
\subsection{Omitted proofs from Section ~\ref{examples}}\label{appendix_sec3}

\newenvironment{myproofL2}{\paragraph{Proof of Lemma \ref{lemma2}.}}{\hfill $\square$}

\begin{myproofL2}
Consider a tournament $T$, on $n$ vertices and a set $S$ of size $k$ such that $S$ can manipulate the matches between them to a tournament $T'$ so that their probability of winning increases by $\alpha_{k,n}^r$, i.e. 
$$\alpha_S^r(T) = r_{S}(T') - r_{S}(T) =  \alpha_{k,n}^r$$
Let $u$ be a new team which loses to all teams in $T$. By Lemma \ref{lemma1} it follows that $r_{S}(T') = r_{S}(T' \cup u)$ and $r_{S}(T) = r_{S}(T \cup u)$. Therefore, 
$$\alpha_{k,n+1}^r \geq r_{S}(T' \cup u) - r_{S}(T \cup u) = r_{S}(T') - r_{S}(T) = \alpha_{k,n}^r$$
\end{myproofL2}

\newenvironment{myproofL3}{\paragraph{Proof of Lemma \ref{lemma3}.}}{\hfill $\square$}

\begin{myproofL3}
Suppose the contrary i.e. there is some $v \in T \setminus S$ which wins in some execution of RDM. Consider the round in which the last team $u \in S$ is eliminated by some team $u'$. Clearly, $u' \notin S$ and thus it cannot eliminate $u$, i.e. a contradiction. 
\end{myproofL3}

\newenvironment{myproofLmonotone}{\paragraph{Proof of Lemma \ref{rdm_monotone}}}{\hfill $\square$}

\begin{myproofLmonotone}
If $u$ beats $v$ in $T'$, then $r_{u}(T) = r_{u}(T')$. Suppose $v$ beats $u$ in $T'$. Consider an execution $E$ of RDM in $T'$ in which $u$ wins. Notice that this execution cannot contain the match $(u,v)$ since $u$ will get eliminated before $v$. Therefore, $E$ is also a valid execution of RDM in $T$. This, provides an injective mapping from the valid executions in $T'$ where $u$ wins to the valid executions in $T$ where $u$ wins. Therefore, $r_{u}(T) \geq r_{u}(T')$.
\end{myproofLmonotone}

\subsection{Omitted proofs from Section ~\ref{rdm3}}\label{appendix_sec4}

\newenvironment{myproof}{\paragraph{Proof of Lemma \ref{lemmaA}}}{\hfill $\square$}

\begin{myproof} We claim that there exists some $x^{*} \in \{u,v,w\}$ such that $\min(\ell_{x^{*}}, \ell'_{x^{*}}) \geq 1$. Indeed the possible values for $(\ell_u,\ell_v,\ell_w)$ and $(\ell'_u,\ell'_v,\ell'_w)$ are $(1,1,1)$ or some permutation of $\{2,1,0\}$. Therefore, as there are least 2 non-zero entries in each of $(\ell_u,\ell_v,\ell_w)$ and $(\ell'_u,\ell'_v,\ell'_w)$, there must be a non-zero entry on which they overlap. This means that there is some $x^{*} \in \{u,v,w\}$ such that $\min(\ell_{x^{*}}, \ell'_{x^{*}}) \geq 1$. Let $\{u,v,w\}\setminus x^{*} = \{y^{*},z^{*}\}$.
By Theorem \ref{theorem1}, we have 
\begin{equation}\label{equation}
r_{\{u,v,w\} \setminus x^{*}}(T' \setminus x^{*}) -  r_{\{u,v,w\}\setminus x^{*}}(T \setminus x^{*}) \leq \frac{1}{3}
\end{equation}
By using the above results and the fact that $\ell'_{x^{*}} + \ell'_{y^{*}} +\ell'_{z^{*}} = 3$ we get the following bound
\begin{align*}
A &= \sum_{x \in \{u,v,w\}} \ell'_x r_{\{u,v,w\} \setminus x}(T' \setminus x) - \ell_x r_{\{u,v,w\}\setminus x}(T \setminus x) \leq \\
&\leq \ell'_{x^{*}} r_{\{u,v,w\} \setminus x^{*}}(T' \setminus x^{*}) - \ell_{x^{*}} r_{\{u,v,w\}\setminus x^{*}}(T \setminus x^{*}) + \ell'_{y^{*}} r_{\{u,v,w\} \setminus y^{*}}(T' \setminus y^{*}) \\
&+ \ell'_{z^{*}} r_{\{u,v,w\} \setminus z^{*}}(T' \setminus z^{*}) \leq  \\
&\leq r_{\{u,v,w\} \setminus x^{*}}(T' \setminus x^{*}) -  r_{\{u,v,w\}\setminus x^{*}}(T \setminus x^{*}) + (\ell'_{x^{*}}-1) r_{\{u,v,w\}\setminus x^{*}}(T' \setminus x^{*})  \\
&+ \ell'_{y^{*}} r_{\{u,v,w\} \setminus y^{*}}(T' \setminus y^{*})+\ell'_{z^{*}} r_{\{u,v,w\} \setminus z^{*}}(T' \setminus z^{*}) \leq \\
&\leq r_{\{u,v,w\} \setminus x^{*}}(T' \setminus x^{*}) -  r_{\{u,v,w\}\setminus x^{*}}(T \setminus x^{*}) + (\ell_{x^{*}}-1) + \ell_{y^{*}}+\ell_{z^{*}} \leq \text{ (by (\ref{equation}))}\\
&\leq \frac{1}{3} + 3-1 = \frac{7}{3}
\end{align*}
where in the second line, we used $\ell_{x^{*}} \geq 1$ to get $-\ell_{x^{*}} r_{\{u,v,w\}\setminus x^{*}}(T \setminus x^{*}) \leq -r_{\{u,v,w\}\setminus x^{*}}(T \setminus x^{*})$ and in the in third line we used $\ell'_{x^{*}} \geq 1$ to get  
$(\ell'_{x^{*}}-1) r_{\{u,v,w\}\setminus x^{*}}(T' \setminus x^{*}) \leq (\ell'_{x^{*}}-1)$. Thus,
$$
    A \leq \frac{7}{3} 
$$
Suppose that $|I| \leq 1$. Then, for every $x \in S$, there is at most one match in which a team in $\{u,v,w\} \setminus x$ loses from a team in $(T \setminus x)  \setminus (S \setminus x)$. Therefore, by Corollary \ref{corollary}, $r_{\{u,v,w\}\setminus x}(T \setminus x) \geq \frac{2}{3}$. Also, we clearly have that $r_{\{u,v,w\}\setminus x}(T' \setminus x) \leq 1$. Thus, 
$$A = \sum_{x \in \{u,v,w\}} \ell'_x r_{\{u,v,w\} \setminus x}(T' \setminus x) - \ell_x r_{\{u,v,w\}\setminus x}(T \setminus x) \leq \ell'_u + \ell'_v + \ell'_w - \frac{2}{3}(\ell_u + \ell_v + \ell_w) = 3-\frac{2}{3} \times 3 = 1$$
Thus, 
$$A \leq 1$$ when $|I| \leq 1$ as desired.
\end{myproof}

\newenvironment{myproofB}{\paragraph{Proof of Lemma \ref{lemmaB}}}{\hfill $\square$}

\begin{myproofB}\label{appendix_B}

By Theorem \ref{theorem1} we have that for all permutations $(x,y,z)$ of $\{u,v,w\}$
$$
    r_{\{z,y\}}(T' \setminus x)- r_{\{z,y\}}(T \setminus x) = r_{\{u,v,w\} \setminus x}(T' \setminus x)- r_{\{u,v,w\} \setminus x}(T \setminus x) \leq \alpha^{RDM}_2 = \frac{1}{3}
$$
By using the above for each of the 3 terms in the sum $B$ we get 
$$
    B = \sum_{x \in S} d^{*}_x (r_{\{u,v,w\} \setminus x}(T' \setminus x)- r_{\{u,v,w\} \setminus x}(T \setminus x)) \leq \frac{d^{*}_u + d^{*}_v + d^{*}_w}{3} = \frac{|I|}{3}
$$
Suppose $|I| \leq 1$. If $|I| = 0$, then $d^{*}_u = d^{*}_v = d^{*}_w = 0$, so $B = 0$. Suppose $|I| = 1$ and WLOG team $a \notin S$ beats $u$ and $\{u,v,w\}$ win all other matches with teams outside of $S$ (in particular $v,w$ beat $t$). Then $d^{*}_v = d^{*}_w = 0$. Thus, 
$$B = d^{*}_u (r_{v,w}(T' \setminus u)- r_{v,w}(T \setminus u)) = 0$$
since $r_{v,w}(T' \setminus u) = r_{v,w}(T \setminus u) = 1$ by Lemma \ref{lemma3}, as $v,w$ beat every team outside of $\{v,w\}$ in both $T' \setminus u$ and $T \setminus u$.
\end{myproofB}

\newenvironment{myproofC}{\paragraph{Proof of Lemma \ref{lemmaC}}}{\hfill $\square$}

\begin{myproofC}\label{appendix_C}
For a tournament $H$, where $u,v,w \in H$, define $t(H) = (a'_1,a'_2,a'_3)$, where in $H$ there are exactly $a'_i$ teams in $H \setminus \{u,v,w\}$ each of which beats exactly $i$ teams out of $\{u,v,w\}$ for $i = 1,2,3$. Suppose that $e \in Q$. Then we have 3 cases for $e$:
\begin{itemize}
    \item Suppose $e$ is a match between two teams from $L_1$ or between some team from $L_1$ and one of the two teams in $\{u,v,w\}$ to which it loses. Then clearly, $t(T|_e) = (a_1-1,a_2,a_3)$ and thus 
$$r_{u,v,w}(T'|_{e})-r_{u,v,w}(T|_{e}) \leq M_{n-1}(a_1-1,a_2,a_3) $$
Notice that there are $\binom{a_1}{2} + 2a_1$ such matches.
\item Suppose now that $e$ is a match between two teams from $L_2$ or between some team from $L_2$ the the unique team in $\{u,v,w\}$ to which it loses. Then, clearly, $t(T|_e) = (a_1,a_2-1,a_3)$. Thus, 
$$r_{u,v,w}(T'|_{e})-r_{u,v,w}(T|_{e}) \leq M_{n-1}(a_1,a_2-1,a_3)$$
There are $\binom{a_2}{2} + a_2$ such matches. 
\item Finally, if $e$ is a match between two teams in $L_3$, then $t(T|_e) = (a_1,a_2,a_3-1)$. Thus,
$$r_{u,v,w}(T'|_{e})-r_{u,v,w}(T|_{e}) \leq M_{n-1}(a_1,a_2,a_3-1)$$
There are $\binom{a_3}{2}$ such matches. 
\end{itemize}
Therefore, by applying the  the above 3 bounds we get 
\begin{align*}
    C &= \sum_{e \notin G} r_{u,v,w}(T'|_{e})-r_{u,v,w}(T|_{e}) \leq \\
    &\leq (2a_1 + \binom{a_1}{2})M_{n-1}(a_1-1,a_2,a_3)+ (a_2 + \binom{a_2}{2}) M_{n-1}(a_1,a_2-1,a_3)\\
    &+ \binom{a_3}{2}M_n(a_1,a_2,a_3-1) +\sum_{e \notin G \cup Q} r_{u,v,w}(T'|_{e})-r_{u,v,w}(T|_{e}))
\end{align*}
as desired. Notice that $|Q| = 2a_1 + a_2 + \binom{a_1}{2} + \binom{a_2}{2} + \binom{a_3}{2}$.
\end{myproofC}
\begin{remark}
Notice that in Lemma \ref{lemmaC} one can obtain an even better bound for $C$ by also considering the matches between two different sets $L_i$ and $L_j$ for $i \neq j$. This could potentially be helpful in getting more precise upper bounds in Table \ref{tab:upp_bounds}. However, for our purposes the above bounds suffice.
\end{remark}

\newenvironment{myprooflI4}{\paragraph{Omitted cases from the proof of Lemma \ref{lemmaI4}\\}}{\hfill $\square$}
\begin{myprooflI4}\label{appendix_cases} 
\textit{Case 4: }$(a_1,a_2,a_3) = (2,0,0)$. Applying $\Delta$, we get 
\begin{align*}
    M_n(2,0,0) &\leq \frac{1}{\binom{n}{2}}(\frac{7}{3} + \frac{2}{3} + (4+ 1) f(1,0,0) + (\binom{n}{2}-10)f(2,0,0)) = \\
    &= \frac{1}{\binom{n}{2}}(3 + \frac{5}{6} + (\binom{n}{2}-10)\frac{23}{60}) = \frac{23}{60} = f(2,0,0)
\end{align*} \\

\textit{Case 5: } $(a_1,a_2,a_3) = (3,0,0)$. Applying $\Delta$, we get
\begin{align*}
M_n(3,0,0) &\leq \frac{1}{\binom{n}{2}}(\frac{7}{3} + \frac{3}{3} + (6+ 3) f(2,0,0) + (\binom{n}{2}-15)f(3,0,0)) = \\
    &= \frac{1}{\binom{n}{2}}(\frac{10}{3} + 9 \times \frac{23}{60} + (\binom{n}{2}-15)\frac{407}{900})=\\
    &= \frac{407}{900} + \frac{1}{\binom{n}{2}}(\frac{200 + 207}{60}-\frac{407}{60}) = \frac{407}{900} = f(3,0,0)
\end{align*}\\

\textit{Case 6} $(a_1,a_2,a_3) = (4,0,0)$. Applying $\Delta$, we get 
\begin{align*}
M_n(4,0,0) &\leq \frac{1}{\binom{n}{2}}(\frac{7}{3} + \frac{4}{3} + (2\times 4+6) f(3,0,0) + (\binom{n}{2}-21)f(4,0,0)) = \\
    &= \frac{1}{\binom{n}{2}}(\frac{11}{3} + 14 \times \frac{407}{900} + (\binom{n}{2}-21)\frac{4499}{9450}) = \\
    &= \frac{4499}{9450} + \frac{1}{\binom{n}{2}}(\frac{3300 + 4098 + 1600}{900}- \frac{8998}{900}) = \frac{4499}{9450} = f(4,0,0)
\end{align*} \\

\textit{Case 7} $(a_1,a_2,a_3) = (0,1,0)$. Applying $\Delta$, we get
\begin{align*}
M_n(0,1,0) &\leq \frac{1}{\binom{n}{2}}(\frac{7}{3} + \frac{2}{3} +  f(0,0,0) + (\binom{n}{2}-6)f(0,1,0)) = \\
    &=  \frac{1}{\binom{n}{2}}(\frac{7}{3} + \frac{2}{3} + (\binom{n}{2}-6)\frac{1}{2}) = \frac{1}{2} = f(0,1,0)
\end{align*}\\

\textit{Case 8} $(a_1,a_2,a_3) = (1,1,0)$. 
Applying $\Delta$, we get
\begin{align*}
M_n(1,1,0) &\leq \frac{1}{\binom{n}{2}}(\frac{7}{3} + \frac{3}{3} +  2f(0,1,0) + f(1,0,0)+ (\binom{n}{2}-9)f(1,1,0)) = \\
    &= \frac{1}{\binom{n}{2}}(\frac{10}{3}  +  \frac{2}{2} + \frac{1}{6}+ (\binom{n}{2}-9)\frac{1}{2})=\\
    &= \frac{1}{2} + \frac{1}{\binom{n}{2}}(\frac{27}{6}-\frac{9}{2}) = \frac{1}{2} = f(1,1,0)
\end{align*}\\

\textit{Case 9} $(a_1,a_2,a_3) = (2,1,0)$. Applying $\Delta$, we get
\begin{align*}
M_n(2,1,0) &\leq \frac{1}{\binom{n}{2}}(\frac{7}{3} + \frac{4}{3} +  (4+1)f(1,1,0) + f(2,0,0)+ (\binom{n}{2}-13)f(2,1,0)) = \\
  &= \frac{1}{\binom{n}{2}}(\frac{11}{3} + \frac{5}{2}+\frac{23}{60}+ (\binom{n}{2}-13)\frac{131}{260}) = \\
  &= \frac{131}{260} + \frac{1}{\binom{n}{2}}(\frac{393}{60}-\frac{131}{20}) = \frac{131}{260} = f(2,1,0)
\end{align*}\\

\textit{Case 10} $(a_1,a_2,a_3) = (0,0,1)$. In this case a vertex $t$ beats $u,v,w$ and they lose to everyone else. By symmetry the probability a team from $u,v,w$ wins is the same no matter what the matches between them are. Thus, $M(0,0,1) = 0 = f(0,0,0)$. \\

\textit{Case 11} $(a_1,a_2,a_3) = (1,0,1)$. Applying $\Delta$, we get
\begin{align*}
M_n(1,0,1) &\leq \frac{1}{\binom{n}{2}}(\frac{7}{3} + \frac{4}{3} +  f(0,0,1)  (\binom{n}{2}-9)f(1,0,1)) = \\
&= \frac{1}{\binom{n}{2}}(\frac{7}{3} + \frac{4}{3} + (\binom{n}{2}-9)\frac{11}{27}) = \frac{11}{27} + \frac{1}{\binom{n}{2}}(\frac{11}{3}-\frac{11}{3}) = \frac{11}{27} = f(1,0,1)\\
\end{align*}

\end{myprooflI4}

\subsection{Omitted proofs from Section \ref{copiessection}}\label{appendix_sec5}

\newenvironment{myproofto13}{\paragraph{Proof of Lemma \ref{lemmacopies}}}{\hfill $\square$}

\begin{myproofto13}
We will construct $q$ and prove the statement by induction on the quantity $|T| + m$. As a base case suppose $|T| + m \leq 3$. In order to define $q$, we need $|T| \geq 2$ and this implies $m = 1$ and $|T| = 2$. In this case there is a single tournament in $T' = T \in Syb(T,u_1,1)$. One can simply define $q(1, T, u_1, v) = r_{v}(T)$. Suppose that we have defined $q(m,T,u_1,v)$ for all $T,u_1 \in T, v \in T \setminus u_1$ and $m$ that satisfy $|T| +m < N$ and that $r_{v}(T') = q(m,T,u_1,v)$ for all $T' \in Syb(T,u_1,m)$ when the parameters satisfy the former constraints.  \\
\indent We will now construct $q$ and show the statement of the Lemma simultaneously for all settings which satisfy $m + |T| = N$. Let $T, u_1 \in T, v \in T \setminus u_1, m \in \mathbb{N}$ satisfy $m + |T| = N$. If $m =1$ the statement is clear because there is a single tournament $T' = T \in Syb(T,u_1,m)$ and thus $q(1,u_1,T,v) = r_v(T) = r_v(T')$. Suppose $m \geq 2$ and $T' \in Syb(T,u_1,m)$. Let $A$ be the set of teams in $T \setminus u_1$ that $u_1$ beats and $B$ the set of teams in $T \setminus u_1$ that $u_1$ loses to. Suppose $|A| = a \geq 0$ and $|B| = b \geq 0$. If the first match is between a team from $B$ and a Sybil or between two Sybils, then a Sybil is eliminated and the resulting tournament is a member of $Syb(T,u_1,m-1)$ (possibly after relabeling $u_1$). By the induction hypothesis, $v$ wins with probability $q(m-1,T,u_1,v)$. There are $\binom{m}{2} +bm$ such matches. If a match in which a team $t \in T \setminus \{v,u_1\}$ loses is played, then the resulting tournament is a member of $Syb(T \setminus t, u_1, m)$. Thus, by the induction hypothesis $v$ wins in the resulting tournament with probability $q(m, T \setminus t,u_1,v)$. If $d_t$ is the number of teams $t$ loses to in $T$, there are $d_t + m-1$ such matches 
if $t \in A$ and $d_t$ such matches if $t \in B$. 
If $v$ loses then, it clearly wins with probability 0. Thus, by first-step analysis we have that 
\begin{align*}
    r_{v}(T') &= \frac{\binom{m}{2} + bm}{\binom{m+a+b}{2}}q(m-1,T,u_1,v) + \sum_{t \in A\setminus v} \frac{d_t + m-1}{\binom{m+a+b}{2}} q(m, T \setminus t, u_1,v) \\
    &+ \sum_{t \in B\setminus v} \frac{d_t}{\binom{m+a+b}{2}} q(m, T \setminus t, u_1,v)
\end{align*}
Notice that $d_t$ depends only on the matches within $T$. Therefore, all terms in the RHS only depend on $T,u_1,m$ and $v$, which finishes the induction and gives us a way to define $q(m,T,u_1,v)$ as the above expression. Also, since $T' \in Syb(T,u_1,m)$ was arbitrary we get $r_{v}(T') = q(m,T,u_1,v)$ for all $T' \in Syb(T,u_1,m)$. 

\end{myproofto13}

\newenvironment{myproofto0}{\paragraph{Proof of Theorem \ref{theoremcopiesto0}}}{\hfill $\square$}

\begin{myproofto0}
Let $T$ be a tournament and $u_1 \in T$ a team, which is not a Condorcet winner. As in Definition 2.10, we need to show that 
$$\lim_{m \to \infty}\max_{T' \in Syb(T,u_1,m)} r_{u_1, \ldots, u_m}(T') = 0$$
Notice that 
$$\max_{T' \in Syb(T,u_1,m)} r_{u_1, \ldots, u_m}(T') = h(m,T,u_1)$$
and that $h(m,T,u_1) \leq p(m,T,u_1)$ so it suffices to show the stronger claim $\lim_{m \to \infty} p(m,T,u_1) = 0$ \\
\indent Recall that $A$ ($B$) denotes the set of teams $u_1$ beats (loses to) in $T$. We prove the statement by induction on $|A|$. If $|A| = 0$, then clearly $T \setminus u_1$ consists of only teams of $B$, which beat $u_1, \ldots,u_m$ and thus, $p(m,T,u_1) = 0$ for all $m$ by Lemma \ref{lemma3}. Let $|A| = a \geq 1$ and $|B| = b \geq 1$ ($u_1$ is not a Condorcet winner). Suppose that the statement is true for all $T$ and $u_1$ with $|A| = a' < a$.  By above $|T| \geq 3$. Let for a team $v$ denote by $d_v$ the number of teams to which it loses in $T$. Similarly to the proof of Lemma \ref{lemmacopies} by first-step analysis we have the following relation 
$$
p(m,T,u_1) = \frac{\binom{m}{2} + bm}{\binom{m+a+b}{2}}p(m-1,T,u_1) + \sum_{v \in B} \frac{d_v}{\binom{n+a+b}{2}} p(m, T \setminus v, u_1) + \sum_{v \in A} \frac{d_v + m-1}{\binom{m+a+b}{2}} p(m, T \setminus v, u) 
$$
 Notice that if $v \in B$, then $d_v \leq a+b$ as $v$ beats the Sybils of $u_1$, $|B| \leq b$, and $p(m, T \setminus v, u_1) \leq 1$. Therefore,
$$\sum_{v \in B} \frac{d_v}{\binom{m+a+b}{2}} p(m, T \setminus v, u_1) \leq \frac{b(a+b)}{\binom{m+a+b}{2}} \leq \frac{C}{m^2}$$
for some $C = C(a,b)$. Also we know that for every $v \in A$, $u_1$ loses to at least one team in $T \setminus v$ (a team in $B$). Therefore, by the inductive hypothesis $\lim_{m \to \infty} p(m, T \setminus v, u_1) = 0$ for all $v \in A$. Let $\epsilon > 0$ and as $A$ does not depend on $m$, we can choose $N$ such that for $m > N$, $p(m, T \setminus v, u_1) < \epsilon$ for all $v \in A$. Also notice that $\sum_{v \in A} d_v + m-1 \leq ma + \binom{a+b+1}{2} \leq Dm$, where $D$ is some constant depending on $a$ and $b$. Thus, collecting the above observations we get for all $n$ big enough that
$$p(m,T,u) \leq \frac{\binom{m}{2} + bm}{\binom{m+a+b}{2}}p(m-1,T,u_1) + \frac{mD}{\binom{m+a+b}{2}}\epsilon + \frac{C}{m^2} $$
By Lemma \ref{decreasing} the function $\frac{\binom{m}{2} + bm}{\binom{m+1+b}{2}}$ is decreasing in $b$ for $b \geq 1$. We use this, fact together with $ a \geq 1$ to show $\frac{\binom{m}{2} + bm}{\binom{m+a+b}{2}} \leq \frac{\binom{m}{2} + m}{\binom{m+2}{2}}$. We also use that $\frac{mD}{\binom{m+a+b}{2}}\epsilon \leq \frac{F}{m}\epsilon$ for some constant $F$. Therefore, by upper bounding the above we get
$$p(m,T,u_1) \leq \frac{\binom{m}{2} + m}{\binom{m+2}{2}}p(m-1,T,u) + \frac{F}{m}\epsilon + \frac{C}{m^2} = \frac{\binom{m+1}{2} }{\binom{m+2}{2}}p(m-1,T,u_1) + \frac{F}{m}\epsilon + \frac{C}{m^2}$$ Therefore, 
$$\binom{m+2}{2} p(m,T,u_1) \leq \binom{m+1}{2}p(m-1,T,u_1) + Sm \epsilon + Q$$ where $S,Q$ are some constants that depend only on the tournament $T$ and the above holds for $m >N$ sufficiently large. Repeating the above inequality $m-N$ times and adding them up as a telescoping sum we get that 
$$\binom{m+2}{2} p(m,T,u_1) \leq \binom{N+1}{2}p(N-1,T,u_1) + S'\epsilon m^2 + Q(m-N)$$
and thus 
$$p(m,T,u_1) \leq  \frac{\binom{N+1}{2}p(N-1,T,u)}{\binom{m+2}{2}} + \frac{S'\epsilon m^2}{\binom{m+2}{2}} + \frac{Q(m-N)}{\binom{m+2}{2}}$$
The first and the third term go to 0 as $m \to \infty$, and the second term goes to $2S'\epsilon$. Therefore, 
$$\lim_{m \to \infty} p(m,T,u_1) \leq 2S'\epsilon$$
Now, letting $\epsilon \to 0$, we get the desired result.
\end{myproofto0}

\begin{lemma}\label{decreasing}
The function $\frac{\binom{m}{2} + bm}{\binom{m+1+b}{2}}$ is decreasing in $b \in \mathbb{Z}$ for $b \geq 1$
\end{lemma}
\begin{proof}
Indeed, one can show that for $b \geq 1$
\begin{align*}
\frac{\binom{m}{2} + bm}{\binom{m+1+b}{2}} &\geq \frac{\binom{m}{2} + (b+1)m}{\binom{m+2+b}{2}} \\
\iff \frac{m-1 + 2b}{(m+b+1)(m+b)} &\geq \frac{m-1 + 2b+2}{(m+2+b)(m+b+1)} \\
\iff (m-1+2b)(m+2+b) &\geq (m+b)(m+2b+1) \\
\iff 2b &\geq 2
\end{align*}
which holds for $b \geq 1$.
\end{proof}




\end{document}